%% file: main.tex
\ifdefined\isarxivversion
\documentclass[11pt]{article}

\usepackage[numbers]{natbib}

\else
\documentclass[conference]{IEEEtran}
\IEEEoverridecommandlockouts

\def\BibTeX{{\rm B\kern-.05em{\sc i\kern-.025em b}\kern-.08em
    T\kern-.1667em\lower.7ex\hbox{E}\kern-.125emX}}
\fi

\usepackage{amsmath}
\usepackage{amsthm}
\usepackage{amssymb}
\usepackage{algorithm}
\usepackage{subfig}
\usepackage{algpseudocode}
\usepackage{graphicx}
\usepackage{grffile}
\usepackage{wrapfig,epsfig}
\usepackage{url}
\usepackage{xcolor}
\usepackage{epstopdf}

\usepackage{bbm}
\usepackage{dsfont}
\allowdisplaybreaks

\ifdefined\isarxivversion

\usepackage{tikz}
\usepackage{hyperref}  %%% arxiv don't allow this.
\hypersetup{colorlinks=true,citecolor=blue,linkcolor=blue} %%% Zhao : maybe we should comment this in submission.
\usetikzlibrary{arrows}
\usepackage[margin=1in]{geometry}

\else
\usepackage{microtype}
\usepackage{hyperref}
\definecolor{mydarkblue}{rgb}{0,0.08,0.45}
\hypersetup{colorlinks=true, citecolor=mydarkblue,linkcolor=mydarkblue}

\fi
\newtheorem{theorem}{Theorem}[section]
\newtheorem{lemma}[theorem]{Lemma}
\newtheorem{definition}[theorem]{Definition}

\newtheorem{proposition}[theorem]{Proposition}
\newtheorem{corollary}[theorem]{Corollary}

\newtheorem{fact}[theorem]{Fact}

\newcommand{\wh}{\widehat}
\newcommand{\wt}{\widetilde}

\newcommand{\R}{\mathbb{R}}

\renewcommand{\varepsilon}{\epsilon}
\renewcommand{\tilde}{\wt}
\renewcommand{\hat}{\wh}

\DeclareMathOperator*{\E}{{\mathbb{E}}}

\newcommand{\HBE}{\mathrm{HBE}}

\newcommand{\PSE}{\mathrm{PSE}}
\newcommand{\Median}{\mathrm{Median}}

\newcommand{\Zhao}[1]{{\color{red}[Zhao: #1]}}
\makeatletter
\newcommand*{\RN}[1]{\expandafter\@slowromancap\romannumeral #1@}
\makeatother

\newcommand{\lianke}[1]{{\color{blue}[Lianke: #1]}}

\usepackage{lineno}

\makeatletter %%%%%%%%%Zhaozhuo: I add this to align the authors
\newcommand{\linebreakand}{%
  \end{@IEEEauthorhalign}
  \hfill\mbox{}\par
  \mbox{}\hfill\begin{@IEEEauthorhalign}
}
\makeatother

\begin{document}

\date{}

\title{Adaptive and Dynamic Multi-Resolution Hashing for Pairwise Summations\thanks{A Preliminary version of this paper is appeared at BigData 2022.}}

\ifdefined\isarxivversion
 \author{
 Lianke Qin\thanks{\texttt{lianke@ucsb.edu}. UCSB.}
 \and 
 Aravind Reddy\thanks{\texttt{aravind.reddy@cs.northwestern.edu} Northwestern University.}
 \and
 Zhao Song\thanks{\texttt{zsong@adobe.com}. Adobe Research.}
 \and
 Zhaozhuo Xu\thanks{\texttt{zx22@rice.edu}. Rice University.}
 \and 
 Danyang Zhuo\thanks{\texttt{danyang@cs.duke.edu}. Duke University.}
}

\else

\author{\IEEEauthorblockN{Lianke Qin}
\IEEEauthorblockA{\textit{Department of Computer Science} \\
\textit{University of California, Santa Barbara}\\
Santa Barbara, CA \\
lianke@ucsb.edu}
\and
\IEEEauthorblockN{Aravind Reddy}
\IEEEauthorblockA{
\textit{Department of Computer Science}\\
\textit{Northwestern University}\\
Evanston, IL \\
aravind.reddy@cs.northwestern.edu}
\and
\IEEEauthorblockN{Zhao Song}
\IEEEauthorblockA{
\textit{Adobe Research}\\
\textit{Adobe}\\
San Jose, CA \\
zsong@adobe.com}
% \and
\linebreakand 
\IEEEauthorblockN{Zhaozhuo Xu}
\IEEEauthorblockA{\textit{Department of Computer Science} \\
\textit{Rice University}\\
Houston, TX \\
zx22@rice.edu}
\and
\IEEEauthorblockN{Danyang Zhuo}
\IEEEauthorblockA{\textit{Department of Computer Science} \\
\textit{Duke University}\\
Durham, NC \\
danyang@cs.duke.edu}
}

\fi

\ifdefined\isarxivversion
\begin{titlepage}
  \maketitle
  \begin{abstract}
\input{abstract}

  \end{abstract}
  \thispagestyle{empty}
\end{titlepage}

%{\hypersetup{linkcolor=black}
%\tableofcontents
%}
%\newpage

\else

\maketitle
\begin{abstract}
\input{abstract}
\end{abstract}

%{\hypersetup{linkcolor=black}
%\tableofcontents
%}

\fi
\input{intro} %%% Section 1. Introduction

\input{prelim}

\input{alg}

\input{multi_resolution_hbe}

\input{adaptive}

\input{conclusion}
%
%\newpage
\ifdefined\isarxivversion
%\section*{Acknowledgments}
\bibliographystyle{alpha}
\bibliography{ref}
\else
\small
%%%Zhao: This is ML First author last name et al style,
% \bibliographystyle{plainnat}
%%%Zhao: This is number style
%\bibliographystyle{plain} %%% Zhao: For this paper, let's use this one
\bibliographystyle{IEEEtran} %%% Zhao: For this paper, let's use this one
%%%Zhao: This is TCS ABC+12 style
% \bibliographystyle{alpha}
\bibliography{ref}

\fi

% \input{checklist}

% \newpage
% \onecolumn
% \appendix
% \section*{Appendix}
% \input{appendix}
% \input{single_appendix}
% \input{multi_appendix}
% \input{missing_proofs}

%%%% Cut-line between first 10 pages and appendix

%\printbibliography[heading=bibintoc,title={References}]
%\section*{References}
%\printbibliography[heading=none]

%%% some writing rules

%% Writing rule for creating tags.
%% Tags :
%% Theorem    \ref{thm:bla_bla}
%% Lemma      \ref{lem:bla_bla}
%% Claim      \ref{cla:bla_bla}
%% Corollary  \ref{cor:bla_bla}
%% Fact       \ref{fac:bla_bla}
%% Definition \ref{def:bla_bla}
%% Section    \ref{sec:bla_bla}
%% Subsection \ref{sub:bla_bla}
%% Equation   \ref{eq:bla_bla}

\end{document}

%% file: abstract.tex
In this paper, we propose Adam-Hash: an adaptive and dynamic multi-resolution hashing data-structure for fast pairwise summation estimation. Given a data-set $X \subset \mathbb{R}^d$, a binary function $f:\mathbb{R}^d\times \mathbb{R}^d\to \mathbb{R}$, and a point $y \in \mathbb{R}^d$, the Pairwise Summation Estimate $\mathrm{PSE}_X(y) := \frac{1}{|X|} \sum_{x \in X} f(x,y)$. For any given data-set $X$, we need to design a data-structure such that given any query point $y \in \mathbb{R}^d$, the data-structure approximately estimates $\mathrm{PSE}_X(y)$ in time that is sub-linear in $|X|$. Prior works on this problem have focused exclusively on the case where the data-set is static, and the queries are independent. In this paper, we design a hashing-based PSE data-structure which works for the more practical \textit{dynamic} setting in which insertions, deletions, and replacements of points are allowed. Moreover, our proposed Adam-Hash is also robust to adaptive PSE queries, where an adversary can choose query $q_j \in \mathbb{R}^d$ depending on the output from previous queries $q_1, q_2, \dots, q_{j-1}$.

%% file: intro.tex
%!TEX root=main.tex
\section{Introduction}

Pairwise Summation Estimation (PSE) is one of the most important problems in machine learning~\cite{sgf95, ks12, j17,ss17, biw19,cs19, chs20,cmf+20,clp+20}, statistics~\cite{ts92, j93, v03, ka05, bng05,kap22}, and scientific computing~\cite{edhd02, oss09, lw19, bc20,hdwy20, hhh+21, cgc+22}. Given a data-set $X \subset \R^d$, a binary function $f:\R^d\times \R^d\to \R$, and a point $y \in \R^d$, we need the pairwise summation estimate of $f(x,y)$ for $x \in X$ i.e. $\PSE_X(y) = \frac{1}{|X|}\sum_{x \in X}f(x,y)$. PSE arises naturally in ML applications: (1) Efficient training and inference of neural network:
In computer vision and natural language processing, the Softmax layer with $n$ neurons is defined as $\frac{\exp{\langle w_i,x \rangle}}{\sum_{j=1}^{n}\langle w_j,x \rangle}$, where $w_i$ is the parameter for the $i$'th neuron and $x$ is the input hidden vector. A novel line of research~\cite{ss17,cmf+20,clp+20} applies PSE to estimate $\sum_{j=1}^{n}\langle w_j,x \rangle$ with running time sublinear in $n$. As a result, a sparse training and inference scheme of Softmax layer can be achieved for acceleration. (2) Fast kernel density estimation: Given a binary kernel function, we would like to estimate the density of a dataset on a query for efficient outlier detection~\cite{llp07,ks12,szk14}, classification~\cite{gcs06,chs20} and clustering~\cite{a09}.

\subsection{The Need for Adaptive and Dynamic PSE}

Recently, there has also been a lot of recent interest in developing PSE for deep learning that are robust to \textit{adaptive} queries \cite{cn20, ssx21,ssx21_rl,clp+20,szz21,cn22,sxz22,gqsw22}, %\Danyang{cite something},
in which an adversary can choose a query $q_j \in \R^d$ that depends on the output of our data structure to past queries $q_1, q_2, \dots, q_{j-1}$.  This is a natural setting in training neural networks. For instance, the input hidden vector of Softmax layer serves as a query for PSE. If we perform adversarial attacks~\cite{pxd+19} in each step, the PSE-accelerated training and inference might lead to failure in generalization. This brings new challenges for PSE because an adversary can always pick the hardest query point, and we would like to have an accuracy guarantee to hold even under this adversarial setting.

Moreover, current PSE applications in ML depend on parameters that are often changing in time and not known apriori. So it is necessary that we develop data structures which can support \textit{dynamic} updates \cite{cls19,lsz19,s19, qgt+19,jklps20,y20,sy21,syz21,jswz21,dly21,qjs+22, hjs+22,gs22,syyz22,hswz22,als+22}.
One typical setting to consider is an iterative process where one data point changes for each iteration. For instance, the weights of Transformer models~\cite{clp+20} changes significantly in the first couple of iterations and then changes smoothly in later iterations. As a result, a successful PSE algorithm for deep models should be robust to incremental updates.

\subsection{Our Proposal: Adam-Hash}

In this paper, we propose Adam-Hash: an  adaptive and dynamic multi-resolution hashing for fast pairwise summation estimation. Formally, we start with defining the problem as follows:

\iffalse
This raises an important question:
\begin{center}
    {\it Can we design a kernel density estimator for a dynamic and adaptive setting?}
\end{center}
\fi

\begin{definition}[Adaptive and dynamic hashing based estimator]\label{def:dynamic_adaptive_hbe}
Given an approximation parameter $\epsilon \in (0,1)$ and a threshold $\tau \in(0,1)$, for every convex function $w$, there exists a data structure which supports the  operations as followed:
\begin{itemize}
    \item \textsc{Initialize}$()$. Initialize the data structure.
    \item \textsc{Query}$(x)$. Given an input $x$, \textsc{Query} outputs an estimate $\hat{Z}$ which approximates $\mu = w(x)$.
     even in a sequence of adaptively chosen queries
    \item \textsc{Update}$(i,z)$. Replace $x_i$ by $z$ into the data structure.
\end{itemize}
\end{definition}

\textbf{Technical contributions.} Adam-Hash contain a dynamic version of 
Hashing-Based-Estimators (HBE) and multi-resolution HBE. 
HBE hashes each data point into a set of hash buckets, and uses collision probability to estimate the pairwise summation function. For updating datasets, we only need to update the corresponding hash buckets. To enable our Adam-Hash to work for adversarial queries, we start from query with a single HBE in Adam-Hash of constant success probability, and boost constant success probability to high probability by obtaining the median of a set of HBE estimators. We then design a $\epsilon_0$-net to prove our Adam-Hash can answer a fixed set of on-net points with high probability. And finally, we generalize the results to all query points where $\| q\|_2 \leq 1$ with the Lipschitz property of the target function.

\paragraph{Roadmap.} 
We introduce some related work to our paper in section~\ref{sec:related_work}.
We give a technique overview of our paper in section~\ref{sec:technique_overview}.
We present the dynamic version $\HBE$ in section~\ref{sec:dynamic_single_resolution_hbe}. Then we present the dynamic version of multi-resolution $\HBE$ in section~\ref{sec:dynamic_multi_hbe}. We further extend to the adaptive and dynamic version $\HBE$ in section~\ref{sec:adaptive_and_dynamic_hbe}. We conclude our contribution in section~\ref{sec:conclusion}.

\section{Related Work}\label{sec:related_work}

\paragraph{Pairwise summation estimation}
The pairwise summation estimation is a general formulation to a lot of machine machine problems. For instance, density estimation of kernel functions~\cite{ss98} is a standard PSE problem. 

Recently, there is an growing direction in using hashing for PSE~\cite{cs17,biw19,cs19,srp+19,kap22}. 
The general intuition is that the binary function is actually a similarity function between data point and query. As a result, we could have speedups in PSE with near neighbor search data structures.

\paragraph{Adaptive queries}
In modern machine learning algorithms that involves data structures. The queries to these data structures are adaptive in two ways: (1) we have sequential queries that is non-i.i.d, such as weights each iteration of the model~\cite{ssx21_rl}, (2) the potential threats posed by deploying algorithms in adversarial settings~\cite{gss14,  pmg16, lcls17, yhzl19}, where an attacker could manipulate query based on the results of previous queries. Thus, current data structures should be robust in these settings so that they can be deployed in machine learning.

%% file: prelim.tex
\section{Technique Overview}\label{sec:technique_overview}

In this section, we present an overview of our techniques to realize Adam-Hash algorithm. Our introduction to Adam-Hash follows a divide-and-conquer style. We start with presenting a dynamic version of multi-resolution hashing. Next, we show how to make it robust to adaptive queries. Finally, we introduce the main algorithm in Algorithm~\ref{alg:adaptive_hbe1}. 

We also provide an overview of our theoretical analysis.
Given an estimator $Z$ of complexity $C$ which is $V$-bounded where $\E[Z] = \mu \in (0,1]$ and threshold $\tau \in (0,1)$,
we first introduce the dynamic version of single-resolution $\HBE$ in Theorem~\ref{thm:single_hbe}.
The query algorithm interacts with the data structure by calling the hash function and sample a data point from the hash bucket the query falls into (Line~\ref{alg:hbe_estimator_compute} of Algorithm~\ref{alg:hbe1}).  After $O\left(V\left((\mu)_{\tau} \log (1 / \delta)\right)\right)$ hash function calls, with probability at least $1-\delta / 2$ we either have a $(1 \pm \epsilon)$-approximation result or \textsc{Query} outputs $0$ indicating that $\mu<\tau$. When we want to update the dataset, we need to insert new entries or delete old entries from each hash table.

Then we present the dynamic version of multi-resolution $\HBE$ in Theorem~\ref{thm:dynamic_multi_hbe}. The estimator consists of $G$ different hashing function families, and each hash function family has a weight computed from their collision probability functions (Line~\ref{alg:multi_hash_weight} of Algorithm~\ref{alg:multi_hbe1}). And the subquery estimation is more complex by combining the estimation from each hash function family at their corresponding weight (Line~\ref{alg:multi_hbe_estimator_compute} of Algorithm~\ref{alg:multi_hbe2}). After interacting with the data structure for $O\left(V\left((\mu)_{\tau} \log (1 / \delta)\right)\right)$ times, with high probability we either have a $(1 \pm \epsilon)$-approximation result or $0$ is outputted indicating that $\mu<\tau$.

To make our data structure adaptive to adversarially chosen queries, we begin with a single HBE estimator to answer the query with $(1 \pm \epsilon)$-approximation with a constant success probability $0.9$.
Then  by chernoff bound we have that obtaining the median of a set of HBE estimators can boost constant success probability to high success probability.
We design a $\epsilon_0$-net $N$ containing $|N| = (10/\epsilon_0)^d$ points and prove that our data structure can answer a fixed set of on-net points with high probability via union bound.
Finally, we know that for each point $q \notin N$, there exists a $p \in N$ such that $\| p - q \|_2 \leq \epsilon_0$. We can quantize the off-net query $q$ to its nearest on-net query $p$ and
generalize the results to all query points where $\| q\|_2 \leq 1$ with the $k$-Lipschitz property of the target function. 
To this end, we complete our proof for designing adaptive and dynamic data structures for multi-resolution hashing of pairwise summation estimates.

\section{Preliminary}\label{sec:prelim}
\paragraph{Notation.}
For a vector $A \in \R^d$, we define $\|A\|_{\infty} = \max_{i \in [d]}(x_i)$.
We define $\| A \|_2 =  \sqrt{\sum_{i=1}^{n} x_i^2}$. We use $[n]$ to denote $\{1,2,\cdots, n\}$. For an event $f(x)$, we define ${\bf 1}_{f(x)}$ such that ${\bf 1}_{f(x)} = 1$ if $f(x)$ holds and ${\bf 1}_{f(x)} = 0$ otherwise. We use $\Pr[\cdot]$ to denote the probability, and use $\E[\cdot]$ to denote the expectation if it exists. We use $a \in (1\pm \epsilon) \cdot b$ to denote $(1 - \epsilon) \cdot b \leq a \leq (1+\epsilon)\cdot b$.

We will make use of Hoeffding's Inequality:
\begin{theorem}[Hoeffding's Inequality~\cite{blm13}]\label{them:hoeffding_inequality}
Let $X_{1}, \ldots, X_{n}$ be independent random variables such that $X_{i} \in[a_{i}, b_{i}]$ almost surely for $i \in[n]$ and let $S=\sum_{i=1}^{n} X_{i}-\mathbb{E}[X_{i}]$. Then, for every $t>0$ :
\begin{align*}
 \Pr[S \geq t] \leq \exp (-\frac{2 t^{2}}{\sum_{i=1}^{n}(b_{i}-a_{i})^{2}}) .
\end{align*}

\end{theorem}

%% file: alg.tex
\section{Dynamic Single-Resolution \texorpdfstring{$\HBE$}{~}}\label{sec:dynamic_single_resolution_hbe}

\iffalse
We present our data structure design in Algorithm~\ref{alg:hbe1} and Algorithm~\ref{alg:hbe2} and the corresponding theorem for the dynamic version $\HBE$ in Theorem~\ref{thm:single_hbe}.
\fi 
In this section, we introduce a dynamic version of single-resolution $\HBE$. 
We present our theorem for the dynamic version $\HBE$ in Theorem~\ref{thm:single_hbe}.

\begin{theorem}[Dynamic version of Theorem 4.2 in page 11~\cite{cs17}]\label{thm:single_hbe}
  For a kernel $w$, given a $V$-bounded estimator $Z$ of complexity $C$ where $\E[Z] = \mu \in (0,1]$ and parameters $\epsilon \in (0,0.1), \tau \in (0,1), \delta \in(0,1)$, there exists a data structure which supports the operations as followed:
  \begin{itemize}
      \item \textsc{Initialize}$(w :\R^d \times \R^d \rightarrow \R_{+}, V :\R \rightarrow \R_{+}, \{x_i\}_{i=1}^{n} \subset \R^d, \mathcal{H}, \epsilon \in (0,0.1), \delta \in (0,1), \tau \in (0,1))$. Given  a set of data points $\{x_i\}_{i=1}^{n}$, a hashing scheme $\mathcal{H}$, the target function $w: \R^d \times \R^d \rightarrow \R_{+}$, the relative variance function $V: \R \rightarrow \R_{+}$, accuracy parameter $\epsilon \in (0,0.1)$, failure probability $\delta \in (0,1)$ and a threshold $\tau \in (0,1)$ as input, the \textsc{Initialize} operation takes $O(\epsilon^{-2} V(\tau) C \log (1 / \delta) \cdot n)$ time.
      \item \textsc{Query}$(x \in \R^d, \alpha \in (0,1], \tau \in (0,1), \delta \in (0,1))$. Given a query point $x \in \R^d$, accuracy parameter $\alpha \in (0,1]$, a threshold $\tau \in (0,1)$ and a failure probability $\delta \in (0,1)$  as input, the time complexity of \textsc{Query} operation is $O(\epsilon^{-2} V(\tau) C \log (1 / \delta) )$ and the output of \textsc{Query} $\hat{Z}$ satisfies:
         \begin{equation*}  
            \left\{
            \begin{aligned}
            &\Pr[|\hat{Z} - \mu| \leq \alpha \mu] \geq 1 - \delta, &\mu \geq \tau  \\
            &\Pr[\hat{Z}  = 0] \geq 1 - \delta, &\mu < \tau 
            \end{aligned}
            \right.
        \end{equation*}  
    %   \lianke{The original AMR algorithm accuracy parameter is $\alpha$. } 
      \item \textsc{InsertX}$(x \in \R^d)$. Given a data point $x \in \R^d$ as input, the \textsc{InsertX} operation takes $O(\epsilon^{-2} V(\tau) \log(1/\delta) \cdot C)$ time to update the data structure.
        \item \textsc{DeleteX}$(x \in \R^d)$. Given a data point $x \in \R^d$ as input, the \textsc{DeleteX} operation takes $O(\epsilon^{-2} V(\tau) \log(1/\delta) \cdot C)$ time to update the data structure. 
  \end{itemize}
\end{theorem}
\begin{proof}
We can prove the theorem by combining the query correctness proof Lemma~\ref{lem:query_correctness} and running time lemmas including the \textsc{Initialize} running time in Lemma~\ref{lem:init_time}, the \textsc{Query} running time in Lemma~\ref{lem:query_time}, the \textsc{InsertX} running time in Lemma~\ref{lem:insert_time} and \textsc{DeleteX} running time in Lemma~\ref{lem:delete_time}.
\end{proof}

We remark that this statement provide a foundation of dynamic $\HBE$s and will help the final presentation of Adam-Hash.

We present our dynamic single-resolution $\HBE$ estimator data structure in Algorithm~\ref{alg:hbe1} and Algorithm~\ref{alg:hbe2}. During \textsc{Initialize}, we evaluate $R$ hash functions on all of the data points to obtain $R$ hash tables. During \textsc{SubQuery} we leverage the hash function collision probability to compute an estimate of the pairwise summation function by 
\begin{align*}
    Z_{i,j} = \frac{1}{n} \frac{w(x,y)}{p_r(x,y)} |H_r(x)| \;\; \forall i \in [L], j \in [m],
\end{align*}
and leverage median of means to output the final estimation. \textsc{Query} calls \textsc{SubQuery} operation for up to $Q = \lfloor \frac{\log(\tau/(1 - (c +\epsilon)))}{\log(1- \gamma)} \rfloor$ times to keep approaching the true result and finally obtain the estimated output. \textsc{InsertX} and \textsc{DeleteX} operations insert or delete corresponding hash table entries using the input data point $x \in \mathcal{S}^{d-1}$.

\iffalse

\subsection{Preliminary Tools on Estimator Complexity}\label{sec:tools_estimator_complexity}
\Zhao{This section looks incremental, needs to ask zhaozhuo if there is a way to make it better}

\begin{theorem}[Exponential kernel, Theorem 5.2  in page 14~\cite{cs17}]
   Let $\beta \in(0,1]$ denote a parameter. Given an exponential kernel $e^{-\|x-y\|_2}$, we can have a   $\HBE$ 
   that is $a(\beta, \sqrt{e})$-scale free and has $T = O\left(d R^{2}\right)$ time complexity in computation.
\end{theorem}

\begin{theorem}[Generalized $t$-student kernel, Theorem 5.3 in page 15~\cite{cs17}]
   For integers $p, q \geqslant 1$, there exists $a(\frac{q}{p}, 3^{q})$-scale free $\HBE$ for the kernel $\frac{1}{1+\|x-y\|_2^{p}}$that has complexity $T=O(d p)$.
\end{theorem}

\begin{theorem}[Gaussian kernel, Theorem 5.4  in page 15~\cite{cs17}]
    Let $t$ denote a value in $[1, R]$, We show that for any $t$, given a Gaussian kernel $e^{-\|x-y\|_2^{2}}$, we can have a $\HBE$  $Z_{t}$ for it that, $\mu^{2} \cdot 4 e^{\frac{3}{2}} \mu^{-\gamma^{2}+\gamma-1}$. Here we write $\gamma(t, \mu):=t / \sqrt{\log (1 / \mu)}$. Moreover, the time for compute this $\HBE$ is 
    $T=O(d t^{2} R^{2})$.
\end{theorem}

\fi

\begin{algorithm}[!h] % enter the algorithm environment
\caption{Dynamic Single-solution $\HBE$-Estimator Data Structure}
\label{alg:hbe1} 
\small
\begin{algorithmic}[1]
\State {\bf data structure}
\State {\bf members}
\State \hspace{4mm} $X = \{x_i\}_{i=1}^{n} \subset \R^d$ \Comment{A set of data points}
\State \hspace{4mm} $R \in \mathbb{N}$ \Comment{Number of hash functions.}
\State \hspace{4mm} $\{h_r\}_{r=1}^{R}$ \Comment{$R$ hash functions}
\State \hspace{4mm} $\{H_r\}_{r=1}^{R}$ \Comment{A collection of hash tables}
\State \hspace{4mm} $\{p_r \}_{r=1}^{R}: \R^d \times \R^d \rightarrow [0,1]$\Comment{The collision probability for hashing schemes}
\State \hspace{4mm} $w :\R^d \times \R^d \rightarrow \R_{+}$ \Comment{The target pairwise function}
\State \hspace{4mm} $V :\R \rightarrow \R_{+}$ \Comment{The relative variance function}
\State {\bf end members}

\Procedure{Initialize}{$w :\R^d \times \R^d \rightarrow \R_{+}, V :\R \rightarrow \R_{+}, \{x_i\}_{i=1}^{n} \subset \R^d, \epsilon \in (0,0.1), \delta \in (0,1), \tau \in (0,1)$} \Comment{Lemma~\ref{lem:init_time}}
\State $R \gets O(\epsilon^{-2}\log(1/\delta)V(\tau))$, $X \gets \{x_i\}_{i=1}^{n}$, $w \gets w$, $V \gets V$
\State Sample $\{h_r\}_{r=1}^{R} \sim v$ from $\mathcal{H}$. $\{p_r\}_{r=1}^{R}$ are corresponding collision probability functions.
\For{$r = 1 \to R$}
    \State $H_r \gets h_r(X)$ \Comment{Evaluate hash function on the dataset to obtain a hash table.}
\EndFor

\EndProcedure

\Procedure{SubQuery}{$x \in \R^d, V : \R \rightarrow \R_{+}, \mu, \epsilon \in (0,1), \delta_0 \in (0,1)$} \Comment{ Lemma~\ref{lem:subquery_correctness}, Lemma~\ref{lem:subquery_time}} 
\State $m \gets \lceil 6\epsilon^{-2}V({\mu}) \rceil$
\State $L \gets \lceil 9 \log(1/\delta_0)\rceil$ 
\For{$i = 1 \to L$}
    \State Sample $r \sim [R]$
        \For{$j = 1 \to m$}
        \State Sample a data point $y \sim H_r(x)$ \Comment{$H_r(x)$ denotes the hash bucket where query $x$ falls into using hash function $h_r$}
        \State $Z_{i,j} \gets \frac{1}{n} \frac{w(x,y)}{p_r(x,y)} |H_r(x)|$ \label{alg:hbe_estimator_compute} \Comment{The hashing-based estimator}
        % \State \Zhao{Use $Z_{i,j}$}
    \EndFor
\EndFor
\State $Z_{i} \gets \text{mean}\{Z_{i,1}, \cdots, Z_{i,m}\}$ for $i \in [L]$
\State $Z \gets \text{median}\{Z_{1}, \cdots, Z_{L}\}$ 
\State \Return $Z$
\EndProcedure
\end{algorithmic}
\end{algorithm}

\begin{algorithm}[!h] % enter the algorithm environment
\caption{Dynamic Single-solution $\HBE$-Estimator Data Structure}
\label{alg:hbe2} 
\small
\begin{algorithmic}[1]
\Procedure{Query}{$x \in \R^d, \alpha \in (0,1], \tau \in (0,1), \delta \in (0,1)$} \Comment{ Lemma~\ref{lem:query_time} and Lemma~\ref{lem:query_correctness}.} 
\State $\epsilon \gets \frac{2}{7}\alpha$, $c \gets \frac{\epsilon}{2}$, $\gamma \gets \frac{\epsilon}{7}$, $\delta_0 \gets \frac{2 \alpha}{49 \log(1/\tau)}$ 
\State $Q \gets \lfloor \frac{\log(\tau/(1 - (c +\epsilon)))}{\log(1- \gamma)} \rfloor$ 
\For{$i = -1 \to Q$}
\State \hspace{4mm} $i \gets i+1$
\State \hspace{4mm} $\mu_i \gets (1 - \gamma)^i$
\State \hspace{4mm} $Z_i \gets \textsc{SubQuery}(x, V, \mu_i, \frac{\epsilon}{3}, \delta_0)$
\If{ $|Z_i - \mu_i| \leq c \cdot \mu_i$ }
    \State \textbf{break;}
\EndIf
\EndFor
\If{$i \leq \frac{49 \log(1/\tau)}{2\alpha}$}
    \State \Return $Z_i$
\Else 
    \State \Return $0$
\EndIf
\EndProcedure

\Procedure{InsertX}{$x \in \R^d$} \Comment{Lemma~\ref{lem:insert_time}}
\State $X = X \cup {x}$
\For{$r = 1 \to R$}
    \State $H_r  \gets H_r \cup \{h_r(x)\} $ \Comment{Insert $x$ to its mapping hash bucket.}
\EndFor
\EndProcedure

\Procedure{DeleteX}{$x \in \R^d$} \Comment{Lemma~\ref{lem:delete_time}}
\State $X = X \setminus {x}$
\For{$r = 1 \to R$}
    \State $H_r  \gets H_r \setminus \{h_r(x)\}$ 
\EndFor
\EndProcedure
\State {\bf end data structure}
\end{algorithmic}
\end{algorithm}

\subsection{Correctness of Query}\label{sec:correct_query_single}

% \Zhao{Add English}
In this subsection, we present the lemmas to prove the correctness of \textsc{SubQuery} and \textsc{Query}.
\begin{lemma}[Correctness of SubQuery]\label{lem:subquery_correctness}
Given an estimator $Z$ of complexity $C$ which is $V$-bounded and $\E[Z] = \mu \in (0,1]$, taking a non-decreasing function $V: \R \rightarrow \R_{+}$, $\mu \in (0,1)$,  a query point $x \in \R^d$, accuracy parameter $\epsilon \in (0,1]$ and a failure probability $\delta_0 \in (0,1)$  as input, the  \textsc{SubQuery} in Algorithm~\ref{alg:hbe1} could obtain an estimation value $Z$, which satisfies:
\begin{align*}
    \Pr[|Z-\mu| \geqslant \epsilon \cdot \mu] \leq \delta_0
\end{align*}
using $O(\epsilon^{-2} V(\mu) \log (1/\delta))$ 
samples.  
\end{lemma}

The correctness of the \textsc{Query} operation is shown as follows.

\begin{lemma}[Correctness of Query]\label{lem:query_correctness}
Given an estimator $Z$ of complexity $C$ which is $V$-bounded and $\E[Z] \in (0,1]$, \textsc{Query}  (in Algorithm~\ref{alg:hbe2}) takes a query point $x \in \R^d$, accuracy parameter $\alpha \in (0,1]$, a threshold $\tau \in (0,1)$ and a failure probability $\delta \in (0,1)$ as inputs, and  outputs $Z_{\mathsf{est}} \in \R$ such that: 
\begin{itemize}
    \item $\Pr[|Z_{\mathsf{est}} - \E[Z]| \leq \alpha \E[Z]] \geq 1 - \delta $ if  $\E[Z] \geq \tau$
    \item $\Pr[Z_{\mathsf{est}}  = 0] \geq 1 - \delta$ if $\E[Z] < \tau $
\end{itemize}
\end{lemma}
\begin{proof}

The query algorithm interacts with $\HBE$ data structure by invoking hash functions. $\HBE$ maintains an index of the most recent hash function invoked and increases it by one after each hash function is evaluated, which ensures that a query never computes the same hash function again so that the data points are independently sampled. When a query arrives, the query algorithm begins with the adaptive mean relaxation algorithm. Here we set $\alpha=1$ and probability be $\delta / 2$. After invoking $O\left(V\left((\E[Z])_{\tau} \log (1 / \delta)\right)\right)$ hash functions  with failure probability at most $\delta / 2$, we can have one of the following result: (1)  $\epsilon$-approximation result, (2) \textsc{Query} outputs $0$ which indicates that $\E[Z]<\tau$. 

For the first scenario, we apply the \textsc{SubQuery} algorithm which have a value that underestimates $\E[Z]$ that invokes $O\left(\epsilon^{-2} \log (1 / \delta) V\left((\E[Z])_{\tau}\right)\right)$
more calls to the hash functions  and resut in an $\epsilon$-approximation result with failure probability at most $\delta$ according to Lemma~\ref{lem:subquery_correctness}. 
For the second scenario, the query algorithm outputs $0$ when $\E[Z] < \tau$.

\end{proof}

\subsection{Running Time}\label{sec:time_single}
In this subsection, we prove the running time of each operation in our data structure, including: \textsc{Initialize}, \textsc{SubQuery}, \textsc{Query}, \textsc{InsertX} and \textsc{DeleteX}.
% \Zhao{Needs to ask English}
\begin{lemma}[Initialize Time]\label{lem:init_time}
Given an estimator $Z$ of complexity $C$ which is $V$-bounded and $n$ data points, the time complexity of \textsc{Initialize} in Algorithm~\ref{alg:hbe1} is $O(\epsilon^{-2} V(\tau) \log(1/\delta) \cdot n C)$.
\end{lemma}
\begin{proof}
During \textsc{Initialize} operation, the running time is dominated by the hash function evaluations on the dataset $\{x_i \}_{i=1}^{n}$. The number of hash functions is $R = O(\epsilon^{-2}\log(1/\delta)V(\tau))$, so the \textsc{Initialize} can be done in $n \cdot R C = O(\epsilon^{-2} V(\tau) \log(1/\delta) \cdot n C)$ time.
\end{proof}
Next, we show the \textsc{SubQuery} operation's running time.

\begin{lemma}[SubQuery Time]\label{lem:subquery_time}
Given an estimator $Z$ of complexity $C$ which is $V$-bounded, the time complexity of \textsc{SubQuery} in Algorithm~\ref{alg:hbe1} is $O(\epsilon^{-2} V((\mu)_{\tau})\log(1/\delta_0)C)$.
\end{lemma}
\begin{proof}
In the \textsc{SubQuery} operation, the running time is dominated by evaluating the $\HBE$ with query point $x \in \R^d$ for $mL =  O(\epsilon^{-2} V_{\mu} \log(1/\delta_0))$ times. Because the complexity of $V$-bounded estimator is $C$, we have that the time complexity of \textsc{SubQuery} is $O(\epsilon^{-2} V((\mu)_{\tau})\log(1/\delta_0)C)$.
\end{proof}

We now move the \textsc{Query} operation's running time.

\begin{lemma}[Query Time]\label{lem:query_time}
Given an estimator $Z$ of complexity $C$ which is $V$-bounded, the time complexity of \textsc{Query} in Algorithm~\ref{alg:hbe2} is $O(\epsilon^{-2} V((\mu)_{\tau})\log(1/\delta)C)$.
\end{lemma}
\begin{proof}
Because in \textsc{Query} operations, \textsc{SubQuery} is called for at most fixed $Q = \lfloor \frac{\log(\tau/(1 - (c +\epsilon)))}{\log(1- \gamma)} \rfloor$ times, the time complexity of \textsc{Query} operation is  $O(\epsilon^{-2} V((\mu)_{\tau})\log(1/\delta)C)$.
\end{proof}

Next, we present the running time for the \textsc{InsertX} operation.

\begin{lemma}[InsertX Time]\label{lem:insert_time}
Given an estimator $Z$ of complexity $C$ which is $V$-bounded, the time complexity of \textsc{InsertX} in Algorithm~\ref{alg:hbe2} is $O(\epsilon^{-2} V(\tau) \log(1/\delta) \cdot C)$.
\end{lemma}
\begin{proof}
During \textsc{InsertX} operation, the running time is dominated by the hash function evaluations on the inserted data point $x \in \R^d$. The number of hash functions is $R = O(\epsilon^{-2}\log(1/\delta)V(\tau))$, so the \textsc{InsertX} can be done in $R C = O(\epsilon^{-2}\log(1/\delta)V(\tau) \cdot C)$ time.
\end{proof}

We present the running time for the \textsc{DeleteX} operation as follows.

\begin{lemma}[DeleteX Time]\label{lem:delete_time}
Given an estimator $Z$ of complexity $C$ which is $V$-bounded, the time complexity of  \textsc{DeleteX} in Algorithm~\ref{alg:hbe2} is $O(\epsilon^{-2} V(\tau) \log(1/\delta) \cdot  C)$.
\end{lemma}
\begin{proof}
During \textsc{DeleteX} operation, the running time is dominated by the hash function evaluations on the to be deleted data point $x \in \R^d$. The number of hash functions is $R = O(\epsilon^{-2}\log(1/\delta)V(\tau))$, so the \textsc{DeleteX} can be done in $R C = O(\epsilon^{-2}\log(1/\delta)V(\tau) \cdot C)$ time.
\end{proof}

%% file: multi_resolution_hbe.tex
\section{Dynamic Multi-Resolution $\HBE$}\label{sec:dynamic_multi_hbe}

In this section, we extend the dynamic correction to multi-resolution hashing.
We present our theorem for the dynamic version multi-resolution $\HBE$ in Theorem~\ref{thm:dynamic_multi_hbe}. %and defer the data structure design and proofs in Appendix~\ref{sec:dynamic_multi_hbe_appendix}.
% \lianke{the input should be $x \in \mathcal{S}^{d-1}$ in the multi-resolution HBE paper.}
\begin{theorem}[Our results, Dynamic version of Theorem 5.4 in page 17~\cite{cs19}]\label{thm:dynamic_multi_hbe}
Given an approximation parameter $\epsilon \in (0,1)$ and a threshold $\tau \in(0,1)$, for a convex function $w$, we have a data structure the allows operations as below:
\begin{itemize}
    \item \textsc{Initialize}$(w :\R^d \times \R^d \rightarrow \R_{+}, V :\R \rightarrow \R_{+}, \{x_i\}_{i=1}^{n} \subset \mathcal{S}^{d-1}, \{ \mathcal{H}_g\}_{g=1}^{G}, \epsilon \in (0,0.1), \delta \in (0,1), \tau \in (0,1))$. Given  a set of data points $\{x_i\}_{i=1}^{n} \subset \mathcal{S}^{d-1}$, a collection of hashing scheme $ \{ \mathcal{H}_g\}_{g=1}^{G}$, the target function $w: \R^d \times \R^d \rightarrow \R_{+}$, the relative variance function $V: \R \rightarrow \R_{+}$, accuracy parameter $\epsilon \in (0,0.1)$, failure probability $\delta \in (0,1)$ and a threshold $\tau \in (0,1)$ as input, the \textsc{Initialize} operation takes $O(\epsilon^{-2} V(\tau) C \log (1 / \delta) \cdot n)$ time. %\lianke{do not replace the complexity of estimator.}
    \item \textsc{Query}$(x \in \mathcal{S}^{d-1}, \alpha \in (0,1], \tau \in (0,1), \delta \in (0,1))$. Given a query point $x \in \mathcal{S}^{d-1}$, accuracy parameter $\alpha \in (0,1]$, a threshold $\tau \in (0,1)$ and a failure probability $\delta \in (0,1)$  as input, the time complexity of \textsc{Query} operation is $O(\epsilon^{-2} V(\tau) C \log (1 / \delta) )$ and the output of \textsc{Query} $\hat{Z}$ satisfies:
         \begin{equation*}  
            \left\{
            \begin{aligned}
            &\Pr[|\hat{Z} - \mu| \leq \alpha \mu] \geq 1 - \delta, &\mu \geq \tau  \\
            &\Pr[\hat{Z}  = 0] \geq 1 - \delta, &\mu < \tau 
            \end{aligned}
            \right.
        \end{equation*}  
    \item \textsc{InsertX}$(x \in \mathcal{S}^{d-1})$. Given a data point $x \in \mathcal{S}^{d-1}$ as input, the \textsc{InsertX} operation takes $O(\epsilon^{-2} V(\tau) \log(1/\delta) \cdot C)$ time to update the data structure.
    \item \textsc{DeleteX}$(x \in \mathcal{S}^{d-1})$. Given a data point $x \in \mathcal{S}^{d-1}$ as input, the \textsc{DeleteX} operation takes $O(\epsilon^{-2} V(\tau) \log(1/\delta) \cdot C)$ time to update the data structure. 
\end{itemize}
\end{theorem}
\begin{proof}
   We can prove the theorem by combining the running time lemmas including the \textsc{Initialize} running time in Lemma~\ref{lem:multi_init_time}, the \textsc{Query} running time in Lemma~\ref{lem:multi_query_time}, the \textsc{InsertX} running time in Lemma~\ref{lem:multi_insert_time} and \textsc{DeleteX} running time in Lemma~\ref{lem:multi_delete_time}, and query correctness proof Lemma~\ref{lem:multi_query_correctness}.
\end{proof}

We remark that as for now, we have a dynamic data structure for $\PSE$ with potential applications in neural network training and kernel density estimation.

We present the dynamic multi-resolution $\HBE$ estimator data structure in Algorithm~\ref{alg:multi_hbe1} and Algorithm~\ref{alg:multi_hbe2}. During \textsc{Initialize}, we evaluate $R$ hash functions for $G$ hash function families on all of the data points to obtain $G \cdot R$ hash tables. During \textsc{SubQuery} we leverage different hash function collision probabilities of $G$ hash function families to compute an estimate of the pairwise summation function by $\forall i \in [L], j \in [m]$
\begin{align*}
    Z_{i,j} = \frac{1}{|X|} \sum_{g =1}^{G} \frac{ \wt{w}_{r,g} (x, y_g) \cdot w (x, y_g)}{p_{r,g}(x, y_g)}|H_{r,g}(x)| %\;\; \forall i \in [L], j \in [m], g \in [G]
\end{align*}
and leverage median of means to output the final estimation.
\textsc{Query} calls \textsc{SubQuery} operation for up to $Q = \lfloor \frac{\log(\tau/(1 - (c +\epsilon)))}{\log(1- \gamma)} \rfloor$ times to keep approaching the true result and finally obtain the estimated output. \textsc{InsertX} and \textsc{DeleteX} operations insert or delete corresponding hash table entries for all hash function families using the input data point $x \in \mathcal{S}^{d-1}$.

\begin{algorithm}[!ht] % enter the algorithm environment
\caption{Multi-Resolution $\HBE$-Estimator Data Structure}
\label{alg:multi_hbe1} 
\small
\begin{algorithmic}[1]
\State {\bf data structure}
\State {\bf members}
\State \hspace{4mm} $X = \{x_i\}_{i=1}^{n} \subset \R^d$ \Comment{A set of data points}
\State \hspace{4mm} $R \in \mathbb{N}$ \Comment{Number of estimators.}
\State \hspace{4mm} $G \in \mathbb{N}$ \Comment{Number of hash schemes.}
\State \hspace{4mm} $\{\mathcal{H}_g\}_{g=1}^{G} $ \Comment{A collection of hashing schemes per estimator. $\mathcal{H}_g = \{f: \R^d \rightarrow \R\}$}
\State \hspace{4mm} $\{\{h_{r,g}\}_{g=1}^{G}\}_{r=1}^{R}$ \Comment{A collection of hash functions}
\State \hspace{4mm} $\{\{H_{r,g}\}_{g=1}^{G}\}_{r=1}^{R}$ \Comment{A collection of hash tables}
\State \hspace{4mm} $\{\{p_{r,g}\}_{g=1}^{G}\}_{r=1}^{R}: \R^d \times \R^d \rightarrow [0,1]$\Comment{The collision probability for hashing functions}
\State \hspace{4mm} $w :\R^d \times \R^d \rightarrow \R_{+}$ \Comment{The target pairwise function}
\State \hspace{4mm} $\{\{\wt w_{r,g}\}_{g=1}^{G}\}_{r=1}^{R}: \R^d \times \R^d \rightarrow \R_+$ \Comment{A collection of weight functions}
\State \hspace{4mm} $V :\R \rightarrow \R_{+}$ \Comment{The relative variance function}
\State {\bf end members}

\Procedure{Initialize}{$w :\R^d \times \R^d \rightarrow \R_{+}, V :\R \rightarrow \R_{+}, \{x_i\}_{i=1}^{n} \subset \R^d, \{ \mathcal{H}_g\}_{g=1}^{G}, \epsilon \in (0,0.1), \delta \in (0,1), \tau \in (0,1)$} \Comment{Lemma~\ref{lem:multi_init_time}}
\State $R \gets O(\epsilon^{-2}\log(1/\delta)V(\tau))$
\State $G \gets \lfloor\frac{\log (\frac{1- | \rho_{+} |}{1 - | \rho_{-} |})}{\log (1+\sqrt{.\frac{\epsilon}{8 |\ell_{\min }}|)}.}\rfloor$ 

\State $X \gets \{x_i\}_{i=1}^{n}$, $w \gets w$, $V \gets V$
\For{$g = 1 \to G$}
    \State Sample $\{h_{r,g}\}_{r=1}^{R} \sim v$ from $\mathcal{H}_g$. $\{p_{r,g}\}_{r=1}^{R}$ are corresponding collision probability functions.
\EndFor
\For{$g = 1 \to G$}
    \For{$r = 1 \to R$}
        \State $H_{r,g} \gets h_{r,g}(X)$ \Comment{Evaluate hash function on the dataset to obtain a hash table.}
        \State ${\wt w_{r,g}}(x,y) \gets \frac{p_{r,g}^2(x,y)}{\sum_{i=1}^{G} p_{r,i}^2(x,y)}$ \label{alg:multi_hash_weight} \Comment{$\sum_{i=1}^{G} {\wt w_{r,i}} = 1$} 
    \EndFor
\EndFor

\EndProcedure
\Procedure{InsertX}{$x \in \mathcal{S}^{d-1}$} \Comment{Lemma~\ref{lem:multi_insert_time}}
\State $X \gets X \cup {x}$
\For{$r = 1 \to R$}
    \For{$g = 1 \to G$}
        \State $H_{r,g}  \gets H_{r,g} \cup \{h_{r,g}(x)\} $ \Comment{Insert $x$ to its mapping hash bucket.}
    \EndFor
\EndFor
\EndProcedure

\Procedure{DeleteX}{$x \in \mathcal{S}^{d-1}$} \Comment{Lemma~\ref{lem:multi_delete_time}}
\State $X \gets X \setminus {x}$
\For{$r = 1 \to R$}
    \For{$g = 1 \to G$}
        \State $H_{r,g}  \gets H_{r,g} \setminus \{h_{r,g}(x)\} $ \Comment{Insert $x$ to its mapping hash bucket.}
    \EndFor
\EndFor
\EndProcedure

\end{algorithmic}
\end{algorithm}

\begin{algorithm}[!h] % enter the algorithm environment
\caption{Multi-Resolution $\HBE$-Estimator Data Structure}
\label{alg:multi_hbe2} 
\small
\begin{algorithmic}[1]
\Procedure{SubQuery}{$x \in \R^d, V : \R \rightarrow \R_{+}, \mu \in (0,1), \epsilon \in (0,1), \delta_0 \in (0,1)$} \Comment{ Lemma~\ref{lem:multi_subquery_correctness}, Lemma~\ref{lem:multi_subquery_time}} 
\State $m \gets \lceil 6\epsilon^{-2}V({\mu}) \rceil$
\State $L \gets \lceil 9 \log(1/\delta_0)\rceil$ 
\For{$i = 1 \to L$}
    \State Sample $r \sim [R]$
        \For{$j = 1 \to m$}
            \For{$g = 1 \to G$}
            \State Sample a data point $y_g \sim H_{r,g}(x)$ \Comment{$H_{r,g}(x)$ denotes the hash bucket where query $x$ falls into using hash function $h_{r,g}$}
            \EndFor
        \State $Z_{i,j} \gets \frac{1}{|X|} \sum_{g =1}^{G} \frac{ \wt{w}_{r,g} (x, y_g) \cdot w (x, y_g)}{p_{r,g}(x, y_g)}|H_{r,g}(x)|$ \label{alg:multi_hbe_estimator_compute} \Comment{The multi-resolution hashing-based estimator}
    \EndFor
\EndFor
\State $Z_{i} \gets \text{mean}\{Z_{i,1}, \cdots, Z_{i,m}\}$ for $i \in [L]$
\State $Z \gets \text{median}\{Z_{1}, \cdots, Z_{L}\}$ 
\State \Return $Z$
\EndProcedure

\Procedure{Query}{$x \in \R^d, \alpha \in (0,1], \tau \in (0,1), \delta \in (0,1)$} \Comment{Lemma~\ref{lem:multi_query_correctness} and Lemma~\ref{lem:multi_query_time}} 
\State $\epsilon \gets \frac{2}{7}\alpha$, $c \gets \frac{\epsilon}{2}$, $\gamma \gets \frac{\epsilon}{7}$, $\delta_0 \gets \frac{2 \alpha}{49 \log(1/\tau)}$
\State $Q \gets \lfloor \frac{\log(\tau/(1 - (c +\epsilon)))}{\log(1- \gamma)} \rfloor$
\For{$i = -1 \to Q$}
\State \hspace{4mm} $i \gets i+1$
\State \hspace{4mm} $\mu_i \gets (1 - \gamma)^i$
\State \hspace{4mm} $Z_i \gets \textsc{SubQuery}(x, V, \mu_i, \frac{\epsilon}{3}, \delta_0)$
\If{ $|Z_i - \mu_i| \leq c \cdot \mu_i$ }
    \State \textbf{break;}
\EndIf
\EndFor
\If{$i \leq \frac{49 \log(1/\tau)}{2\alpha}$}
    \State \Return $Z_i$
\Else 
    \State \Return $0$
\EndIf
\EndProcedure

\State {\bf end data structure}
\end{algorithmic}
\end{algorithm}

\subsection{Correctness of Query}\label{sec:correctness_multi_query}
In this section, we present the lemmas for correctness of \textsc{SubQuery} and \textsc{Query} operation in the dynamic multi-resoluion $\HBE$ in Lemma~\ref{lem:multi_query_correctness}.
\begin{lemma}[Correctness of SubQuery]\label{lem:multi_subquery_correctness}
Given a  multi-resolution $\HBE$ estimator $Z$ of complexity $C$ which is $V$-bounded and  $\E[Z] = \mu \in (0,1]$, taking a non-decreasing function $V: \R \rightarrow \R_{+}$, $\mu \in (0,1)$,  a query point $x \in \R^d$, accuracy parameter $\epsilon \in (0,1]$ and a failure probability $\delta_0 \in (0,1)$  as input, the  \textsc{SubQuery} in Algorithm~\ref{alg:multi_hbe2} can get an estimate $Z$ such that 
\begin{align*}
    \Pr[|Z-\mu| \geqslant \epsilon \cdot \mu] \leq \delta_0
\end{align*}
using $O(\epsilon^{-2} V(\mu) \log (1/\delta))$ samples.  
\end{lemma}
We now move to the correctness for the \textsc{Query} operation.

% \lianke{do not need to write proof for this subquery.}
\begin{lemma}[Correctness of Query]\label{lem:multi_query_correctness}
Given a  multi-resolution $\HBE$ estimator $Z$ of complexity $C$ which is $V$-bounded and  $\E[Z] \in (0,1]$, \textsc{Query} $Z_{\mathsf{est}}$ (in Algorithm~\ref{alg:multi_hbe2}) takes a query point $x \in \R^d$, accuracy parameter $\alpha \in (0,1]$, a threshold $\tau \in (0,1)$ and a failure probability $\delta \in (0,1)$ as inputs, and  outputs $Z_{\mathsf{est}} \in \R$ such that: 
\begin{itemize}
    \item $\Pr[|Z_{\mathsf{est}} - \E[Z]| \leq \alpha \E[Z]] \geq 1 - \delta $ if  $\E[Z] \geq \tau$
    \item $\Pr[Z_{\mathsf{est}}  = 0] \geq 1 - \delta$ if $\E[Z] < \tau $
\end{itemize}
\end{lemma}
\begin{proof}
The query algorithm invokes hash function calls to sample data points from the data structure. The data structure increases the index of the most recent hash function invoked by one after each hash function is computed, which enables a query to never evaluate the same hash function again and sample data points independently. When a query arrives, the query algorithm begins with the adaptive mean relaxation algorithm with $\alpha=1$ and probability $\delta / 2$. After interacting with the data structure for $O\left(V\left((\E[Z])_{\tau} \log (1 / \delta)\right)\right)$ times,  we either obtain an $\epsilon$-approximation result or \textsc{Query} outputs $0$ which indicates that $\E[Z]<\tau$ with probability at least $1-\delta / 2$.

For the first scenario, we apply the \textsc{SubQuery} algorithm which results in an value that underestimates  of $\E[Z]$. Moreover it invokes $O\left(\epsilon^{-2} \log (1 / \delta) V\left((\E[Z])_{\tau}\right)\right)$
more hash function calls and obtains an $\epsilon$-approximation result with success probability of at least $1-\delta$ according to Lemma~\ref{lem:subquery_correctness}. 
For the second scenario, the query algorithm outputs $0$ when $\E[Z] < \tau$.

\end{proof}

\subsection{Running Time}\label{sec:time_multi_hbe}
In this section, we prove the time complexity of \textsc{Initialize}, \textsc{SubQuery}, \textsc{InsertX} and \textsc{DeleteX} operations.
\begin{lemma}[Initialize Time]\label{lem:multi_init_time}
Given a $V$-bounded  multi-resolution $\HBE$ estimator of complexity $C$ and $n$ data points, the time complexity of \textsc{Initialize} in Algorithm~\ref{alg:multi_hbe1} is $O(\epsilon^{-2} V(\tau) \log(1/\delta) \cdot n C)$.
\end{lemma}
\begin{proof}
During \textsc{Initialize} operation, the running time is dominated by the hash function evaluations on the dataset $\{x_i \}_{i=1}^{n}$. The number of hash functions is $R = O(\epsilon^{-2}\log(1/\delta)V(\tau))$, so the \textsc{Initialize} can be done in $n \cdot R C = O(\epsilon^{-2} V(\tau) \log(1/\delta) \cdot n C)$ time.
\end{proof}
The \textsc{SubQuery} time for our data structure is the following.

\begin{lemma}[SubQuery Time]\label{lem:multi_subquery_time}
Given a $V$-bounded multi-resolution $\HBE$ estimator of complexity $C$, the time complexity of \textsc{SubQuery} in Algorithm~\ref{alg:multi_hbe2} is $O(\epsilon^{-2} V((\mu)_{\tau})\log(1/\delta_0)C)$.
\end{lemma}
\begin{proof}
In the \textsc{SubQuery} operation, the running time is dominated by computing the multi-resolution $\HBE$ with query point $x \in \R^d$ for $mL =  O(\epsilon^{-2} V_{\mu} \log(1/\delta_0))$ times. Because the complexity of $V$-bounded estimator is $C$, we have that the time complexity of \textsc{SubQuery} is $O(\epsilon^{-2} V((\mu)_{\tau})\log(1/\delta_0)C)$.
\end{proof}

We present the \textsc{Query} time as follows.

\begin{lemma}[Query Time]\label{lem:multi_query_time}
Given a $V$-bounded multi-resolution $\HBE$ estimator of complexity $C$, the time complexity of \textsc{Query} in Algorithm~\ref{alg:multi_hbe2} is $O(\epsilon^{-2} V((\mu)_{\tau})\log(1/\delta)C)$.
\end{lemma}
\begin{proof}
Because in \textsc{Query} operations, \textsc{SubQuery} is called for at most fixed $Q = \lfloor \frac{\log(\tau/(1 - (c +\epsilon)))}{\log(1- \gamma)} \rfloor$ times, the time complexity of \textsc{Query} operation is  $O(\epsilon^{-2} V((\mu)_{\tau})\log(1/\delta)C)$.
\end{proof}

The \textsc{Insert} time for the data structure is the following.

\begin{lemma}[Insert Time]\label{lem:multi_insert_time}
Given a $V$-bounded multi-resolution $\HBE$ estimator of complexity $C$, the time complexity of \textsc{InsertX} in Algorithm~\ref{alg:multi_hbe1} is $O(\epsilon^{-2} V(\tau) \log(1/\delta) \cdot C)$.
\end{lemma}
\begin{proof}
During \textsc{InsertX} operation, the running time is dominated by the hash function evaluations on the inserted data point $x \in \R^d$. The number of hash functions is $R = O(\epsilon^{-2}\log(1/\delta)V(\tau))$, so the \textsc{InsertX} can be done in $R C = O(\epsilon^{-2}\log(1/\delta)V(\tau) \cdot C)$ time.
\end{proof}

Now we present the \textsc{Delete} time for the data structure.

\begin{lemma}[Delete Time]\label{lem:multi_delete_time}
Given a $V$-bounded multi-resolution $\HBE$ estimator of complexity $C$, the time complexity of  \textsc{DeleteX} in Algorithm~\ref{alg:multi_hbe1} is $O(\epsilon^{-2} V(\tau) \log(1/\delta) \cdot  C)$.
\end{lemma}
\begin{proof}
During \textsc{DeleteX} operation, the running time is dominated by the hash function evaluations on the to be deleted data point $x \in \R^d$. The number of hash functions is $R = O(\epsilon^{-2}\log(1/\delta)V(\tau))$, so the \textsc{DeleteX} can be done in $R C = O(\epsilon^{-2}\log(1/\delta)V(\tau) \cdot C)$ time.
\end{proof}

%% file: adaptive.tex
\section{Adam-Hash: Adaptive and Dynamic $\HBE$}\label{sec:adaptive_and_dynamic_hbe}
In this section, we present Adam-Hash: an adaptive and dynamic multi-resolution hashing-based estimator data structure design in Algorithm~\ref{alg:adaptive_hbe1}, and give the corresponding theorem in Theorem~\ref{thm:dynamic_adaptive_hbe}. Then we present the lemmas and proofs for correctness of query in section~\ref{sec:correctness_query_adaptive}. We present the lemmas and proofs for the time complexity of operations in our dynamic multi-resolution $\HBE$ data structure in section~\ref{sec:time_adaptive_hbe}.

\begin{algorithm}[!h] % enter the algorithm environment
\caption{Adam-Hash Data Structure}
\label{alg:adaptive_hbe1} 
\small
\begin{algorithmic}[1]
\State {\bf data structure}
\State {\bf members}
\State \hspace{4mm} $\{\HBE_j \}_{j=1}^{L}$ \Comment{A set of $\HBE$ estimators}
\State \hspace{4mm} $L \in \mathbb{N}$ \Comment{Number of $\HBE$ estimators}
\State \hspace{4mm} $n \in \mathbb{N}$ \Comment{Number of data points.}
\State {\bf end members}

\Procedure{Initialize}{$w :\R^d \times \R^d \rightarrow \R_{+}, V :\R \rightarrow \R_{+}, \{x_i\}_{i=1}^{n} \subset \R^d, \{ \mathcal{H}_g\}_{g=1}^{G}, \epsilon \in (0,0.1), \delta \in (0,1), \tau \in (0,1)$} \Comment{Lemma~\ref{lem:adaptive_multi_init_time}}
\State $L \gets O(\log((10 k/\epsilon \tau )^d/\delta))$
\State $n \gets n$
\For{$j = 1 \to L$}
    \State $\HBE_j.\textsc{Initialize}(w, V, \{x_i\}_{i=1}^{n}, \{ \mathcal{H}_g\}_{g=1}^{G}, \epsilon, \delta, \tau)$
\EndFor
\EndProcedure
\Procedure{Query}{$q \in \R^d, \epsilon \in (0,1], \tau \in (0,1), \delta \in (0,1)$} \Comment{Lemma~\ref{lem:adaptive_multi_query_correctness}, Lemma~\ref{lem:adaptive_multi_query_time}}
\State  Let $N$ denote the $\epsilon_0$-net of $\{ x \in \R^d ~|~ \| x \|_2 \leq 1 \}$.
\State  Find a point $p \in N$ which is the closest to $q$.
\For{$k \in [L]$}
    \State $y_k = \HBE_{k}.\textsc{Query}(p, \epsilon, \tau, \delta)$
\EndFor
\State $\tilde{z} \gets \text{Median}(\{y_k\}_{k=1}^{L})$
\State \Return $\tilde{z}$
\EndProcedure
\Procedure{InsertX}{$x \in \R^d$} \Comment{Lemma~\ref{lem:adaptive_multi_insert_time}}
\State $n \gets n+1$
\For{$k = 1 \to L$}
    \State $\HBE_j.\textsc{InsertX}(x)$ \Comment{Insert the data point into each $\HBE$ data structure.}
\EndFor
\EndProcedure

\Procedure{DeleteX}{$x \in \R^d$} \Comment{Lemma~\ref{lem:adaptive_multi_delete_time}}
\State $n \gets n-1$
\For{$k = 1 \to L$}
    \State $\HBE_j.\textsc{DeleteX}(x)$ 
\EndFor
\EndProcedure

\end{algorithmic}
\end{algorithm}

\begin{theorem}[Our results, Adaptive and dynamic hashing based estimator]\label{thm:dynamic_adaptive_hbe}
Let the approximation parameter be $\epsilon \in (0,1)$ and a threshold $\tau \in(0,1)$. Given any convex function $w$, we show that we can have a data structure allows operations as below:
\begin{itemize}
    \item \textsc{Initialize}$(w :\R^d \times \R^d \rightarrow \R_{+}, V :\R \rightarrow \R_{+}, \{x_i\}_{i=1}^{n} \subset \mathcal{S}^{d-1}, \{ \mathcal{H}_g\}_{g=1}^{G}, \epsilon \in (0,0.1), \delta \in (0,1), \tau \in (0,1))$. Given  a set of data points $\{x_i\}_{i=1}^{n} \subset \mathcal{S}^{d-1}$, a collection of hashing scheme $ \{ \mathcal{H}_g\}_{g=1}^{G}$, the $k$-Lipschitz target function $w: \R^d \times \R^d \rightarrow \R_{+}$, the relative variance function $V: \R \rightarrow \R_{+}$, accuracy parameter $\epsilon \in (0,0.1)$, failure probability $\delta \in (0,1)$ and a threshold $\tau \in (0,1)$ as input, the \textsc{Initialize} operation takes $O(\epsilon^{-2} V(\tau) C \log (1 / \delta) \cdot n \cdot  \log((10 k/\epsilon \tau )^d/\delta))$ time. 
    \item \textsc{Query}$(x \in \mathcal{S}^{d-1}, \epsilon \in (0,0.1), \tau \in (0,1), \delta \in (0,1))$. Given a query point $x \in \mathcal{S}^{d-1}$, accuracy parameter $\epsilon \in (0,0.1)$, a threshold $\tau \in (0,1)$ and a failure probability $\delta \in (0,1)$  as input, the time complexity of \textsc{Query} operation is $O(\epsilon^{-2} V(\tau) C \log (1 / \delta) \cdot  \log((10 k/\epsilon \tau )^d/\delta))$ and the output of \textsc{Query} $\hat{Z}$ satisfies:
         \begin{equation*}  
            \left\{
            \begin{aligned}
            &\Pr[|\hat{Z} - \mu| \leq \epsilon \mu] \geq 1 - \delta, &\mu \geq \tau  \\
            &\Pr[\hat{Z}  = 0] \geq 1 - \delta, &\mu < \tau 
            \end{aligned}
            \right.
        \end{equation*}  
     even when the queries are adaptive.
    \item \textsc{InsertX}$(x \in \mathcal{S}^{d-1})$. Given a data point $x \in \mathcal{S}^{d-1}$ as input, the \textsc{InsertX} operation takes $O(\epsilon^{-2} V(\tau) \log(1/\delta) \cdot C \cdot  \log((10 k/\epsilon \tau )^d/\delta))$ time to update the data structure.
    \item \textsc{DeleteX}$(x \in \mathcal{S}^{d-1})$. Given a data point $x \in \mathcal{S}^{d-1}$ as input, the \textsc{DeleteX} operation takes $O(\epsilon^{-2} V(\tau) \log(1/\delta) \cdot C \cdot  \log((10 k/\epsilon \tau )^d/\delta))$ time to update the data structure. 
\end{itemize}
\end{theorem}
\begin{proof}
   We can prove the theorem by combining the running time lemmas including Lemma~\ref{lem:adaptive_multi_init_time}, Lemma~\ref{lem:adaptive_multi_query_time}, Lemma~\ref{lem:adaptive_multi_insert_time} and Lemma~\ref{lem:adaptive_multi_delete_time}, and query correctness proof Lemma~\ref{lem:adaptive_multi_query_correctness}.
\end{proof}

\subsection{Correctness of Query}\label{sec:correctness_query_adaptive}
The goal of this section is to prove the correctness for \textsc{Query} in  Algorithm~\ref{alg:adaptive_hbe1} in Lemma~\ref{lem:adaptive_multi_query_correctness}, 
\begin{lemma}[Correctness of Query]\label{lem:adaptive_multi_query_correctness}
Given an estimator $Z$ of complexity $C$ which is $V$-bounded and $\E[Z] = \mu \in (0,1]$, \textsc{Query} $\hat{Z}$ in Algorithm~\ref{alg:adaptive_hbe1} takes a query point $x \in \R^d$, accuracy parameter $\epsilon \in (0,1]$, a threshold $\tau \in (0,1)$ and a failure probability $\delta \in (0,1)$ as inputs, and  outputs $\hat{Z} \in \R$ such that: 
\begin{itemize}
    \item $\Pr[|\hat{Z} - \mu| \leq \epsilon \mu] \geq 1 - \delta $ if  $\mu \geq \tau$
    \item $\Pr[\hat{Z}  = 0] \geq 1 - \delta$ if $\mu < \tau $
\end{itemize}
\end{lemma}
\begin{proof}

When $\mu \geq \tau$, for \textsc{Query} in Algorithm~\ref{alg:adaptive_hbe1}, first we prove each $\HBE$ estimator can answer the query with constant success probability in Lemma~\ref{lem:single_estimator}, second we prove the median of query results from $L$ $\HBE$ estimators can achieve $\epsilon$ approximation with high probability in Lemma~\ref{lem:fixed_points}, third we prove that for all query points on a $\epsilon$-net can be answered with $\epsilon$ approximation with high probability in Lemma~\ref{lem:net_points} and fourth we prove that for all query points $\| q \|_2 \leq 1$, \textsc{Query} in Algorithm~\ref{alg:adaptive_hbe1} can give an answer with $\epsilon$ approximation with high probability in Lemma~\ref{lem:all_points}.

When $\mu < \tau$, the \textsc{Query} in Algorithm~\ref{alg:adaptive_hbe1} returns $0$ with probability $1-\delta$.

\end{proof}

\subsubsection{Starting with Constant Probability}
First we need to prove the $\HBE$ estimator can answer the query approximately with a constant success probability.
\begin{lemma}[Constant probability]\label{lem:single_estimator}
Given $\epsilon \in (0,0.1), \tau \in (0,1)$, a query point $q \in \R^d$ and a set of data points $X = \{ x_i\}_{i}^{n} \subset \R^d$, let 
$
%\begin{align*}
    Z(q) := \frac{1}{|X|} \sum_{x \in X} w(x,q)
%\end{align*}
$
an estimator $\HBE$ can answer the query which satisfies:
\begin{align*}
 \HBE.\textsc{query}(q, \epsilon, 0.1) \in (1 \pm \epsilon)\cdot Z(q)
\end{align*}
with probability $0.9$.
\begin{align*}
\end{align*}
\end{lemma}

\vspace{-4mm}
\subsubsection{Boost the Constant Probability to High Probability}
Then we want to boost the constant success probability to high success probability via obtaining the median of $L$ queries.
\begin{lemma}[Boost the probability]\label{lem:fixed_points}
We write the failure probability as $\delta_1 \in (0,0.1)$ and accuracy parameter as $\epsilon \in (0,0.1)$.
Given $L = O( \log(1/\delta_1) )$  estimators $\{\HBE_j\}_{j=1}^{L}$. For each fixed query point $q \in \R^d$, the median of queries from $L$ estimators satisfies that:
\begin{align*}
   % (1-\epsilon) \cdot Z(q) \leq
    \Median(\{\HBE_j.\textsc{query}(q, \epsilon,  0.1)\}_{j=1}^{L}) \in (1 \pm \epsilon)\cdot Z(q)
\end{align*}
with probability $1 - \delta_1$.
\end{lemma}
\begin{proof}
From Lemma~\ref{lem:single_estimator} we know each estimator $\HBE_j$ can answer the query that satisfies:
\begin{align*}
\HBE.\textsc{query}(q, \epsilon, 0.1) \in (1 \pm \epsilon)\cdot Z(q)
\end{align*}
with probability $0.9$.

From the chernoff bound we know the median of $L =O( \log(1/\delta_1))$ queries from $\{\HBE_j\}_{j=1}^{L}$ satisfies:
\begin{align*}
 \Median(\{\HBE_j.\textsc{query}(q, \epsilon, 0.1)\}_{j=1}^{L}) \in (1 \pm \epsilon)\cdot Z(q)
\end{align*}
with probability $1 - \delta_1$.

Therefore, we complete the proof.
\end{proof}

\subsubsection{From each Fixed Point to All the Net Points}
In this section, we present Fact~\ref{fac:number_of_net_points} and  generalize from each fixed point to all on-net points in Lemma~\ref{lem:net_points}.

\begin{fact}\label{fac:number_of_net_points}
Let $N$ be the $\epsilon_0$-net of $\{ x \in \R^d ~|~ \| x \|_2 \leq 1 \}$. Let $|N|$ denote the size of $N$. Then $|N|\leq (10/\epsilon_0)^d$.
\end{fact}

\begin{lemma}[From for each fixed to for all net points]\label{lem:net_points}
Let $N$ denote the $\epsilon_0$-net of $\{ x \in \R^d ~|~ \| x \|_2 \leq 1 \}$. Let $|N|$ denote size of $N$. Given $L = \log(|N|/\delta)$  estimators $\{\HBE_j\}_{j=1}^{L}$. 

With probability $1 - \delta$, we have: for all $q \in N$, the median of queries from $L$ estimators satisfies that:
\begin{align*}
 \Median(\{\HBE_j.\textsc{query}(q, \epsilon,  0.1)\}_{j=1}^{L}) \in (1 \pm \epsilon)\cdot Z(q).
\end{align*}
\end{lemma}
\begin{proof}
There are $|N|$ points on the $d$ dimension $\epsilon_0$-net when $\| q\|_2 \leq 1$. From Lemma~\ref{lem:fixed_points} we know that for each query point $q$ on $N$, we have : 
\begin{align*}
  \Median(\{\HBE_j.\textsc{query}(q, \epsilon,  0.1)\}_{j=1}^{L}) \in (1 \pm \epsilon)\cdot Z(q)
\end{align*}
with failure probability $\delta/|N|$.

Next, we could union bound all $|N|$ points on $N$ and obtain the following:
\begin{align*}
    &~\forall \| q\|_2 \leq 1 : \\ 
    &~\Median(\{\HBE_j.\textsc{query}(q, \epsilon,  0.1)\}_{j=1}^{L}) \in (1 \pm \epsilon)\cdot Z(q)
\end{align*}
with probability $1 - \delta$.
\end{proof}

\subsubsection{From Net Points to All Points}
With Lemma~\ref{lem:net_points}, we are ready to prove all query points $\| q \|_2 \leq 1$ can be answered approximately with high probability.
\iffalse
\begin{definition}
We say $\HBE$ is $\mathrm{Lip}$-Lipschitz if
\begin{align*}
    | \HBE.\textsc{query}(q_1, \epsilon, \delta) -  \HBE.\textsc{query}(q_2, \epsilon, \delta) | \leq \mathrm{Lip} \cdot \| q_1 - q_2 \|_2
\end{align*}
\end{definition}
\fi
% \lianke{Below lemma is new.}
% \lianke{change $k$ to $L$}
\begin{lemma}[$k$-Lipschitz]
Given a dataset $X \subset \R^d$, and a query vector $q \in \R^d$, the $k$-Lipschitz target function $w(q, x)$, then the summation $Z(q) = \frac{1}{|X|}\sum_{x \in X} w(q, x)$ is $k$-Lipschitz.
\end{lemma}

\begin{proof}
Based on the assumption that the target function $w(q,x)$ is $k$-Lipschitz, we have:
\begin{align}\label{eq:w_lipshitz}
    \forall x \in X: |w(q_1, x) - w(q_2, x)| \leq k \cdot \| q_1 - q_2\|_2
\end{align}
To prove $Z(q)$ is $k$-Lipschitz, we have:
\begin{align*}
    &~ |Z(q_1) - Z(q_2)| \\
    = &~ |\frac{1}{|X|}\sum_{x \in X} w(q_1, x)  - \frac{1}{|X|}\sum_{x \in X} w(q_2, x)| \\
    = &~ |\frac{1}{|X|}\sum_{x \in X} (w(q_1, x) - w(q_2, x))| \\
    \leq &~ \frac{1}{|X|}\sum_{x \in X} ( k \cdot \| q_1 - q_2 \|_2) \\
    = &~ k \cdot \| q_1 - q_2 \|_2
\end{align*}
where the first step comes from the definition of $Z(q)$, the second step comes from merging the summation, the third step comes from Eq.~\eqref{eq:w_lipshitz}. Therefore, we have that $Z(q)$ is $k$-Lipschitz.

\end{proof}

The following fact shows that if the elements of an array shift by a bounded value $\epsilon$, then the median of the array shifts by a bounded value $3 \epsilon$.

\begin{fact}[folklore]\label{fac:median}
Given two list of numbers such that $|a_i - b_i| \leq \epsilon$, for all $i \in [n]$. Then we have
\begin{align*}
   | \Median( \{ a_i\}_{i \in [n]}  ) - \Median( \{ b_i\}_{i \in [n]}  ) | \leq 3\epsilon
\end{align*}
\end{fact}

\begin{lemma}[From net points to all points]\label{lem:all_points}
Given $L = O(\log((10 k/\epsilon \tau )^d/\delta))$  estimators $\{\HBE_j\}_{j=1}^{L}$, with probability $1 - \delta$, for all query points $\|q\|_2 \leq 1$, there exits a point $p \in N$ which is the closest to $q$,  we have the median of queries from $L$ estimators satisfies that:
\begin{align*}
    &~\forall \| q\|_2 \leq 1: \\
    &~ \Median(\{\HBE_j.\textsc{query}(p, \epsilon, 0.1)\}_{j=1}^{L}) \in (1 \pm \epsilon)\cdot Z(q).
\end{align*}
\end{lemma}
\begin{proof}

We define an event $\xi$ to be the following,
\begin{align*}
    &~\forall p \in N,  \\ 
    &~\Median(\{\HBE_j.\textsc{query}(p, \epsilon, 0.1)\}_{j=1}^{L}) \in (1 \pm \epsilon)\cdot Z(p)
\end{align*}

Using Lemma~\ref{lem:net_points} with $L= \log(|N|/\delta)$, we know that 
\begin{align*}
    \Pr[~\text{~event~} \xi \text{~holds~} ] \geq 1 - \delta
\end{align*}
Using Fact~\ref{fac:number_of_net_points}, we know that
\begin{align*}
    L
    =  \log(|N|/\delta ) 
    =  \log( ( 10 /\epsilon_0 )^d / \delta )  
    = & ~ \log((10 k/\epsilon \tau )^d/\delta)
\end{align*}
where the last step follows from $\epsilon_0 \leq \epsilon \tau / k$.

We condition the above event $\xi$ to be held. (Then the remaining proof is not depending on any randomness, for each and for all becomes same.)
For each point $q \notin N$, there exists a $p \in N$ such that
\begin{align}\label{eq:diff_p_q}
    \| p - q \|_2 \leq \epsilon_0
\end{align}
For each $q \notin N$, we quantize off-net query $q$ to its nearest on-net query $p$. we know
\begin{align}\label{eq:diff_z_q_p}
    |Z(q) - Z(p)| \leq  k \cdot \| q - p\|_2  
                  \leq  k \epsilon_0 
                  \leq  \epsilon \tau
\end{align}
where the first step follows that $Z(\cdot)$ is $k$-Lipschitz, the second step comes from Eq.~\eqref{eq:diff_p_q} and the third step comes from $\epsilon_0 \leq \epsilon \tau / k$.
Using the on-net query $p$ to answer the off-net quantized query $q$, we have:
\begin{align*}
     \Median(\{\HBE_j.\textsc{query}(p, \epsilon, 0.1)\}_{j=1}^{L}) \in (1 \pm \epsilon)\cdot Z(p) .
\end{align*}
 Using  Eq.~\eqref{eq:diff_z_q_p}, we can obtain that
\begin{align*}
     \Median(\{\HBE_j.\textsc{query}(p, \epsilon, 0.1)\}_{j=1}^{L}) \in (1 \pm \epsilon)\cdot (Z(q) \pm \epsilon \tau).
\end{align*}
Using $ \forall j \in [L]: \HBE_j.\textsc{query}(p, \epsilon, 0.1) \geq \tau$, we have 
\begin{align*}
   \Median(\{\HBE_j.\textsc{query}(p, \epsilon, 0.1)\}_{j=1}^{L}) \in (1 \pm \epsilon)^2 Z(q) .
\end{align*}
We know that $(1- \epsilon)^2 \geq (1 - 3\epsilon)$ and $(1 + \epsilon)^2 \leq (1 + 3\epsilon)$ when $\epsilon \in (0, 0.1)$. Rescaling the $\epsilon$ completes the proof.
\end{proof}

\subsection{Running Time}\label{sec:time_adaptive_hbe}
In this section, we provide several lemmas to prove the time complexity of each operation in our data structure. We remark that the running time complexity corresponds to potential energy consumption in practice. We hope that our guidance in theory would help improves the energy efficiency in the running of our algorithm.

\subsubsection{Initialization Time}
\begin{lemma}[Initialize Time]\label{lem:adaptive_multi_init_time}
Given an estimator $Z$ of complexity $C$ which is $V$-bounded and $n$ data points, the time complexity of \textsc{Initialize} in Algorithm~\ref{alg:multi_hbe1} is $O(\epsilon^{-2} V(\tau) \log(1/\delta) \cdot n C)$.
\end{lemma}
\begin{proof}
Because in \textsc{Initialize} operation in Algorithm~\ref{alg:adaptive_hbe1}, $\HBE$.\textsc{Initialize} is called for $L = O(\log((10 k/\epsilon \tau )^d/\delta))$ times to initialize $L$ $\HBE$ data structures, and each $\HBE$.\textsc{Initialize} takes $O(\epsilon^{-2} V(\tau) C \log (1 / \delta) \cdot n)$ time to complete. Therefore, the overall time complexity of  \textsc{Initialize} operation in Algorithm~\ref{alg:adaptive_hbe1} is $O(\epsilon^{-2} V(\tau) C \log (1 / \delta) \cdot n \cdot \log((10 k/\epsilon \tau )^d/\delta))$
\end{proof}
\subsubsection{Query Time}
We prove the time complexity of \textsc{Query} in Lemma~\ref{lem:adaptive_multi_query_time}.
\begin{lemma}[Query Time]\label{lem:adaptive_multi_query_time}
Given an estimator $Z$ of complexity $C$ which is $V$-bounded, the time complexity of \textsc{Query} in Algorithm~\ref{alg:adaptive_hbe1} is $O(\epsilon^{-2} V((\mu)_{\tau})\log(1/\delta)C \cdot \log((10 k/\epsilon \tau )^d/\delta))$.
\end{lemma}
\begin{proof}
Because in \textsc{Query} operation in Algorithm~\ref{alg:adaptive_hbe1}, $\HBE$.\textsc{Query} is called for $L = O(\log((10 k/\epsilon \tau )^d/\delta))$ times, and each $\HBE$.\textsc{Query} takes $O(\epsilon^{-2} V(\tau) C \log (1 / \delta))$ time to complete. As a result, the total time of  \textsc{Query} operation in Algorithm~\ref{alg:adaptive_hbe1} is $O(\epsilon^{-2} V((\mu)_{\tau})\log(1/\delta)C \cdot \log((10 k/\epsilon \tau )^d/\delta))$.
\end{proof}
\subsubsection{Maintenance Time}
We provide the time complexity of \textsc{InsertX} in Lemma~\ref{lem:adaptive_multi_insert_time}.% and delay the proof to Appendix~\ref{sec:adaptive_multi_insert_time_app}.
\begin{lemma}[Insert Time]\label{lem:adaptive_multi_insert_time}
Given an estimator $Z$ of complexity $C$ which is $V$-bounded, the time complexity of \textsc{InsertX} in Algorithm~\ref{alg:adaptive_hbe1} is $O(\epsilon^{-2} V(\tau) \log(1/\delta) \cdot C \cdot \log((10 k/\epsilon \tau )^d/\delta))$.
\end{lemma}
\begin{proof}

Because in \textsc{InsertX} operation in Algorithm~\ref{alg:adaptive_hbe1}, $\HBE$.\textsc{InsertX} is called for $L = O(\log((10 k/\epsilon \tau )^d/\delta))$ times, and each $\HBE$.\textsc{InsertX} takes $O(\epsilon^{-2} V(\tau) \log(1/\delta) \cdot C)$ time to complete. Therefore, the overall time complexity of  \textsc{InsertX} operation in Algorithm~\ref{alg:adaptive_hbe1} is $O(\epsilon^{-2} V(\tau) \log(1/\delta) \cdot C \cdot \log((10 k/\epsilon \tau )^d/\delta))$.
\end{proof}

We provide the time complexity of \textsc{DeleteX} in Lemma~\ref{lem:adaptive_multi_delete_time}.% and delay the proof to Appendix~\ref{sec:adaptive_multi_delete_time_app}.
\begin{lemma}[Delete Time]\label{lem:adaptive_multi_delete_time}
Given an estimator $Z$ of complexity $C$ which is $V$-bounded, the time complexity of  \textsc{DeleteX} in Algorithm~\ref{alg:adaptive_hbe1} is $O(\epsilon^{-2} V(\tau) \log(1/\delta) \cdot C \cdot \log((10 k/\epsilon \tau )^d/\delta))$.
\end{lemma}

\begin{proof}
Because in \textsc{DeleteX} operation in Algorithm~\ref{alg:adaptive_hbe1}, $\HBE$.\textsc{DeleteX} operation is called for $L = O(\log((10 k/\epsilon \tau )^d/\delta))$ times, and each $\HBE$.\textsc{DeleteX} operation takes $O(\epsilon^{-2} V(\tau) \log(1/\delta) \cdot C)$ time to complete. Therefore, the overall time complexity of  \textsc{DeleteX} operation operation in Algorithm~\ref{alg:adaptive_hbe1} is $O(\epsilon^{-2} V(\tau) \log(1/\delta) \cdot C \cdot \log((10 k/\epsilon \tau )^d/\delta))$.
\end{proof}

%% file: conclusion.tex
\section{Conclusion}\label{sec:conclusion}

% \lianke{Prof. Zhuo. Could you write a conclusion here? Thanks!}
Pairwise Summation Estimation~(PSE) is an important yet challenging task in machine learning. In this paper, we present Adam-Hash: the first provable adaptive and dynamic multi-resolution hashing  for PSE. In an iterative process, the data set changes by a single data point per iteration, and our data structure outputs an approximation of the pairwise summation of a binary function in sub-linear time in the size of the data set. Our data structure also works for the adaptive setting where an adversary can choose a query based on previous query results. We hope our proposal would shed lights on joint innovations of data structures and machine learning.  